\documentclass[11pt,a4paper]{article}
\usepackage{fullpage}
\usepackage[utf8]{inputenc}
\usepackage{amssymb}
\usepackage[leqno]{amsmath}
\makeatletter
\newcommand{\leqnomode}{\tagsleft@true\let\veqno\@@leqno}
\newcommand{\reqnomode}{\tagsleft@false\let\veqno\@@eqno}
\makeatother
\usepackage{mathabx}
\usepackage{mathtools}
\usepackage{csquotes}
\usepackage[all]{nowidow}
\usepackage{amsthm,amsmath,amssymb}
\usepackage[mathcal,mathscr]{eucal}
\usepackage{epsfig,graphicx,graphics,color}
\usepackage{caption}
\usepackage{subcaption}
\usepackage[dvipsnames]{xcolor} %Notes
\usepackage{tcolorbox}
\usepackage{enumerate}
\usepackage[sort,nocompress]{cite}
\usepackage{hyperref}
\hypersetup{
    colorlinks=true,       % false: boxed links; true: colored links
    linkcolor=blue,        % color of internal links (change box color with linkbordercolor)
    citecolor=red,         % color of links to bibliography
    filecolor=magenta,     % color of file links
    urlcolor=cyan,         % color of external links
    linktocpage=true
}

\usepackage{algorithm}
\usepackage{bbm}
\usepackage[noend]{algpseudocode}
\usepackage{varwidth}
\newcounter{algsubstate}


\theoremstyle{plain}
\newtheorem{theorem}{Theorem}
\newtheorem{lemma}[theorem]{Lemma}
\newtheorem{corollary}[theorem]{Corollary}

\theoremstyle{definition}

\newtheorem{remark}[theorem]{Remark}

\newcommand{\cB}{\mathcal{B}}

\newcommand{\cH}{\mathcal{H}}
\newcommand{\cK}{\mathcal{K}}

\def\oa{\overline{a}}
\def\ob{\overline{b}}
\def\oc{\overline{c}}

\def\bx{\mathbf{x}}

\DeclareMathOperator*{\poly}{poly}
\DeclareMathOperator{\minl}{min'l\;}

\def\final{0}  % set this to 1 to get a comment-free version
\ifnum\final=0  %namely if we allow comments in the output
\newcommand{\krnote}[1]{{\color{red}[{\tiny Krist\'of: \bf #1}]\marginpar{\color{red}*}}}
\newcommand{\onote}[1]{{\color{OliveGreen}[{\tiny \textbf{Ondrej:} \bf #1}]\marginpar{\color{OliveGreen}*}}}
\newcommand{\pnote}[1]{{\color{Mahogany}[{\tiny Petr: \bf #1}]\marginpar{\color{Mahogany}*}}}
\newcommand{\enote}[1]{{\color{NavyBlue}[{\tiny Endre: \bf #1}]\marginpar{\color{NavyBlue}*}}}
\newcommand{\knote}[1]{{\color{Orange}[{\tiny Kaz: \bf #1}]\marginpar{\color{Orange}*}}}
\else % in this case [final=1] we don't want any comments to show
\newcommand{\krnote}[1]{}
\newcommand{\onote}[1]{}
\newcommand{\pnote}[1]{}
\newcommand{\enote}[1]{}
\newcommand{\knote}[1]{}
\fi

\title{Unique key Horn functions}

\author{Krist{\'o}f B{\'e}rczi\thanks{MTA-ELTE Egerv\'ary Research Group, Department of Operations Research, E{\"o}tv{\"o}s Lor{\'a}nd University, Budapest, Hungary. Email: \texttt{berkri@cs.elte.hu}.} 
\and
Endre Boros\thanks{MSIS Department and RUTCOR, Rutgers University, New Jersey, USA. Email: \texttt{endre.boros@rutgers.edu}.}
\and
Ond\v{r}ej \v{C}epek\thanks{Charles University, Faculty of Mathematics and Physics, Department of Theoretical Computer Science and Mathematical Logic, Praha, Czech Republic. Email: \texttt{cepek@ktiml.mff.cuni.cz}.}
\and
Petr Ku\v{c}era\thanks{Charles University, Faculty of Mathematics and Physics, Department of Theoretical Computer Science and Mathematical Logic, Praha, Czech Republic. Email: \texttt{kucerap@ktiml.mff.cuni.cz}.}
\and
Kazuhisa Makino\thanks{Research Institute for Mathematical Sciences (RIMS) Kyoto University, Kyoto, Japan. Email: \texttt{makino@kurims.kyoto.ac.jp}.}
}

\begin{document}
\date{}
\maketitle

\begin{abstract}
Given a relational database, a key is a set of attributes such that a value assignment to this set uniquely determines the values of all other attributes. The database uniquely defines a pure Horn function $h$, representing the functional dependencies. If the knowledge of the attribute values in set $A$ determines the value for attribute $v$, then $A\rightarrow v$ is an implicate of $h$. If $K$ is a key of the database, then $K\rightarrow v$ is an implicate of $h$ for all attributes $v$. 

Keys of small sizes play a crucial role in various problems. We present structural and complexity results on the set of minimal keys of pure Horn functions. We characterize Sperner hypergraphs for which there is a unique pure Horn function with the given hypergraph as the set of minimal keys. Furthermore, we show that recognizing such hypergraphs is co-$\mathsf{NP}$-complete already when every hyperedge has size two. On the positive side, we identify several classes of graphs for which the recognition problem can be decided in polynomial time.

We also present an algorithm that generates the minimal keys of a pure Horn function with polynomial delay. By establishing a connection between keys and target sets, our approach can be used to generate all minimal target sets with polynomial delay when the thresholds are bounded by a constant. As a byproduct, our proof shows that the {\sc Minimum Key} problem is at least as hard as the {\sc Minimum Target Set Selection} problem with bounded thresholds. 
  
  \bigskip

\end{abstract}

\section{Introduction}
\label{sec:intro}
Throughout the paper, we denote by $V$ the set of $n$ Boolean variables. We will refer to the members of $V$ as \emph{positive} and to their negations as \emph{negative literals}, respectively. A \emph{Boolean function} is a mapping $f:\{0,1\}^V\to\{0,1\}$. A \emph{conjunctive normal form} (CNF) is the conjunction of clauses, where each clause is a disjunction of literals. The CNF $\Phi=C_1\wedge\dots\wedge C_q$ is also viewed as a set of clauses $\Phi=\{C_1,\dots,C_q\}$. 

A CNF $\Phi$ is called \emph{Horn} if each of its clauses contains at most one positive literal, and \emph{pure Horn} if every clause contains exactly one positive literal. A Boolean function $h$ is called \emph{(pure) Horn} if it has a (pure) Horn CNF representation. Note that every CNF defines a Boolean function, but a Boolean function may have many different CNF representations. For instance, given the pure Horn CNF $\Phi=(\oa \vee b) \wedge (\ob \vee a) \wedge (\oa \vee \oc \vee d) \wedge (\oa \vee \oc \vee e)$ on variables $a,b,c,d,e$, we can also represent the same Boolean function $h$ by the pure Horn CNF $\Psi=(\oa \vee b) \wedge (\ob \vee a) \wedge (\ob \vee \oc \vee d) \wedge (\ob \vee \oc \vee e)$. Note that a pure Horn clause can also be viewed as an implication. For instance, $C=\ob \vee \oc \vee e$ is equivalent with the implication $bc\rightarrow e$. Thus, we can view a pure Horn CNF as an implication system, e.g., we shall write $\Phi$ equivalently, as $a \rightarrow b,  b \rightarrow a,  ac \rightarrow de$. 
For an implication $A\rightarrow v$ we call $A$ the \emph{body} and $v$ the \emph{head}. We say that $A\rightarrow v$ is an \emph{implicate} of the Horn function $h$ if any assignment $x\in \{0,1\}^V$ that falsifies $A\rightarrow v$ also falsifies $h$. In particular, if $h$ is represented by a pure Horn CNF then the clauses of this CNF are all implicates of $h$. 

The concept of Horn functions has been widely studied under different names, such as directed hypergraphs in graph theory and combinatorics \cite{ADS86}, as implication systems in machine learning \cite{AB09,AB11}, database theory \cite{A74,maier}, and as lattices and closure systems in algebra and concept lattice analysis \cite{CASPARD2003241,GD86}. Horn functions form a fundamental subclass of Boolean functions endowed with interesting structural and computational properties. 
The satisfiability problem can be solved for Horn functions in linear time and the equivalence of such formulas can be decided in polynomial time \cite{DOWLING1984267}. Horn functions are strongly related to relational databases \cite{A74} and many interesting algorithmic problems arise from that context. Given a database, we associate the set $V$ of Boolean variables to the set of attributes of the database. For $A\subseteq V$ and $v\in V$ we write $A\rightarrow v$ if the knowledge of the attribute values in $A$ uniquely determines the value of $v$ (in the database records). Such a relation is called a functional dependency in the database. The set of all functional dependencies define a unique pure Horn function associated to the given database \cite{A74,maier}.  
One of the important notions that arise from databases is the notion of a key. A key in a relational database is a set of attributes the values of which determine uniquely the values of all other attributes. 
Accordingly, a subset $K$ of the variables is a \emph{key} of a Horn function $h$ if $K\rightarrow v$ is an implicate of $h$ for all $v\in V\setminus K$. 

We call a pure Horn function \emph{key Horn} if the body of any of its implicates is a key of the function. 
Key Horn functions generalize the well studied class of \emph{hydra functions}  introduced in \cite{Sloan2017HydrasDH}, where all the bodies are of size $2$. Finding the shortest CNF representation of a given Horn function with respect to multiple relevant measures (number of clauses, number of literals, etc.) is an outstanding hard problem \cite{ADS86,maier,HK93}. For general pure Horn functions not even non-trivial approximation algorithms are known. For hydra functions a 2-approximation algorithm was given in \cite{Sloan2017HydrasDH}, while \cite{Kucera2014HydrasCO} proved that the minimization remains NP-hard even in this special case. In \cite{arxiv}, the authors provided logarithmic factor approximation algorithms for general key Horn functions with respect to all of the above mentioned measures.

\paragraph{Our results}

The present paper focuses on the structure of the set of minimal keys of a pure Horn function. In particular, we are interested in finding Sperner hypergraphs $\cB$ that form the set of minimal keys of a unique pure Horn function $h_\cB$. We call such a $\cB$ a \emph{unique key hypergraph}, and the corresponding Horn function $h_\cB$ a \emph{unique key Horn function}. 

Section~\ref{sec:ukhf} gives a characterization of unique key hypergraphs and unique key Horn functions. In particular, we show that cuts of a matroid form a unique key hypergraph. The special case when every hyperedge has size two is discussed in Section~\ref{sec:ukg}, where we show that recognizing unique key graphs is co-$\mathsf{NP}$-complete. Subsequently, we identify several classes of graphs for which the recognition problem can be decided in polynomial time. Section~\ref{sec:lex} provides an algorithm that generates all minimal keys of a pure Horn function with polynomial delay. Furthermore, we show that the problems of finding a minimum key of a pure Horn function and of finding a minimum target set of a graph are closely related. Using this connection, our algorithm can be used to generate all minimal target sets with polynomial delay when the thresholds are bounded by a constant.

\section{Unique Key Horn Functions}
\label{sec:ukhf}

The purpose of this section is to give an understanding of the structure of pure Horn functions that have the same set of keys and in particular the structure of unique key Horn functions.

We start with additional definitions and notation. We view the set of variables $V$ as a ground set. A hypergraph $\cB\subseteq 2^V$ is called a \emph{Sperner hypergraph} if none of its hyperedges contains another one. Given a Sperner hypergraph $\cB\subseteq 2^V$, we say that $T\subseteq V$ is a \emph{transversal} of $\cB$, if $T\cap B\neq \emptyset$ for all $B\in \cB$. We say that $S$ is an \emph{independent set} of $\cB$ if $T=V\setminus S$ is a transversal of $\cB$. We denote by $\cB^d$ the set of \emph{minimal transversversals} of $\cB$, and by $\cB^*$ the \emph{family of its independent sets}. 

For a hypergraph $\cB\subseteq 2^V$ and subset $S\subseteq V$ we denote by 
$\cB_S = \{B\in \cB\mid B\subseteq S\}$
the \emph{subhypergraph of $\cB$ induced by $S$}. In particular, if $S\in \cB^*$ then $\cB_S=\emptyset$. 
Furthermore, we denote by $\cB^S = \minl \{S\cap B\mid B\in\cB\}$
the \emph{projection of $\cB$ to $S$} where $\minl\{\cH\}$ denotes the family consisting of the inclusionwise minimal members of $\cH$. Clearly, if $S$ is not a transversal of $\cB$ then we have $\cB^S =\{\emptyset\}$. We introduce the notation $\cup \cB$ to denote the union of the hyperedges of $\cB$, i.e. $\cup \cB = \bigcup_{B\in\cB} B$.
We will use the following well-known lemma.

\begin{lemma}[Seymour \cite{seymour1973incomparable}]\label{l1}
For a Sperner hypergraph $\cB\subseteq 2^V$ and subset $S\subseteq V$ we have 
$(\cB_S)^d=(\cB^d)^S$ and $(\cB^S)^d=(\cB^d)_S.$
\end{lemma}

For a Boolean function $h$, we denote by $T(h)$ the set of true vectors of $h$, i.e., $T(h)=\{\bx\in\{0,1\}^V\mid h(\bx)=1\}$. For two functions $h$ and $h'$ we write $h\leq h'$ if for all $\bx\in\{0,1\}^V$ we have $h(\bx)\leq h'(\bx)$, in other words, if $T(h)\subseteq T(h')$. 
We say that a clause $A\to v=v\vee\bigvee_{a\in A} \bar{a}$ is an \emph{implicate} of $h$ if $(A\to v)(\bx)\geq h(\bx)$ for all $\bx \in \{0,1\}^V$. 
For a subset $S\subseteq V$ we define the \emph{forward chaining closure} of $S$ by 
$F_h(S) = \{u\in V\mid S\to u \text{ is an implicate of } h\}$. Note that if $h'\leq h$, then $F_{h'}(S)\supseteq F_h(S)$, since any implicate of $h$ is also an implicate of $h'$. For a CNF $\Phi$ we use the same terminology and notation as it defines a unique Boolean function. For example, $\Phi\subseteq\Psi$ implies $\Phi\geq\Psi$.

Keys of a pure Horn function clearly form an upward monotone system. We denote by $\cK(h)$ the \emph{set of minimal keys} of $h$. 
To a Sperner hypergraph $\cB\subseteq 2^V$ we associate the pure Horn CNF
\[
\Phi_\cB ~=~ \bigwedge_{B\in\cB} \bigwedge_{v\in V\setminus B} (B\to v).
\]
Note that we have $\cK(\Phi_\cB)=\cB$. For a Sperner family $\cB$ we call $\Phi_\cB$ a \emph{key Horn CNF}. Consequently, a pure Horn function is key Horn if and only if it has such a CNF representation. 
Let us observe that for a Sperner hypergraph $\cB$ and pure Horn function $h$, $\cB\subseteq \cK(h)$ implies that $h\leq\Phi_\cB$. 

Let us also note that there may be several pure Horn functions with the same family of keys.  As an example, consider the hypergraph $\cB=\{\{a,b\},\{b,c\},\{c,d\}\}$ over the ground set $V=\{a,b,c,d\}$, and the pure Horn CNFs $\Psi^1=\Phi_\cB\wedge (b\to d)$, $\Psi^2=\Phi_\cB\wedge (c\to a)$, and $\Psi^3=\Phi_\cB\wedge (b\to d)\wedge(c\to a)$.
It is easy to verify now that the CNFs $\Phi_\cB$, and $\Psi^i$, $i=1,2,3$ define four pairwise distinct pure Horn functions and each has $\cB$ as its set of minimal keys.

\begin{lemma}\label{t1}
Let $\cB\subseteq 2^V$ be a Sperner hypergraph and $h:\{0,1\}^V\to \{0,1\}$ be a pure Horn function such that $h\leq \Phi_\cB$. Then $\cK(h)\neq \cB$ if and only if there exists an implicate $A\to v$ of $h$ and a minimal transversal $T\in \cB^d$ such that $T\cap A=\emptyset$ and $v\in T$.
\end{lemma}
\begin{proof}
Since $h\leq \Phi_\cB$, any $B\in\cB$ is a key of $h$. Thus $\cK(h)\subseteq \cB$ implies $\cK(h)=\cB$ because $\cB$ is a Sperner hypergraph. 

Assume first that $\cK(h)\neq \cB$, that is, there exists a minimal key $K\in \cK(h)\setminus \cB$. 
Since the sets of $\cB$ are keys of $h$ and $K$ is a minimal key of $h$, we must have $K\in\cB^*$. Let \(W\) denote a
maximal independent set which contains \(K\) as a subset. It follows
that \(T=V\setminus W\) is a minimal transversal which is disjoint
from \(K\). Let \(v\) be an
arbitrary node in \(T\). Then
\(K\to v\) is an implicate of \(h\) because \(K\)
is a key.  Thus, choosing $A=K$ proves one direction of our claim. 

For the reverse direction, let us assume that there exists an implicate $A\to v$ of $h$ and a minimal transversal $T\in \cB^d$ such that  $T\cap A=\emptyset$ and $v\in T$. Since $T$ is a minimal transversal of $\cB$, there exists $B\in \cB$ such that $T\cap B=\{v\}$. This implies that $F_{(A\to v)\wedge\Phi_\cB}(V\setminus T)=V$. Because we have $h\leq (A\to v)\wedge\Phi_\cB$ by our assumptions, $F_h(V\setminus T)=V$ follows. Therefore there exists a minimal key $K\subseteq V\setminus T$ of $h$. Finally, $V\setminus T\in\cB^*$ implies $K\in\cB^*$, from which $K\in \cK(h)\setminus \cB$ follows as claimed.
\end{proof}

\begin{lemma}\label{c1}
Let $\cB\subseteq 2^V$ be a Sperner hypergraph and $h:\{0,1\}^V\to \{0,1\}$ be a pure Horn function such that $h\leq \Phi_\cB$. Then $\cK(h)=\cB$ if and only if for all implicates $A\to v$ of $h$ with $A\in\cB^*$ we have $v\in (V\setminus A)\setminus (\cup \cB^{V\setminus A})$. 
\end{lemma}
\begin{proof}
Let us first note that for any subset $A\subseteq V$ that has a disjoint minimal transversal $T\in\cB^d$ we must have $A\in\cB^*$. 
Thus, by Lemma \ref{t1}, we have $\cK(h)=\cB$ if and only if for all implicates $A\to v$ of $h$ for which $A\in\cB^*$ and for all minimal transversals $T\in\cB^d$ with $T\cap A=\emptyset$ we have $v\not\in T$. 
Since by Lemma~\ref{l1} we have $(\cB^{V\setminus A})^d=(\cB^d)_{V\setminus A}$ and for all Sperner hyperhraphs $\cH$ the equality $\cup\cH = \cup \cH^d$ holds, we have
$\cup (\cB^d)_{V\setminus A} = \cup \cB^{V\setminus A}$, implying the claim. 
\end{proof}

\begin{lemma}\label{c1b}
Let $\cB\subseteq 2^V$ be a Sperner hypergraph and define $\Psi=\{A\to v\mid A\in\cB^*,v\not\in\cup \cB^{V\setminus A}\}$. Let $\varphi$ be a set of clauses of the form $A\to v$ that are not implicates of $\Phi_\cB$. Then $\cK(\varphi\wedge\Phi_\cB)=\cB$ if and only if $\varphi\subseteq\Psi$. 
\end{lemma}
\begin{proof}
The claim follows by Lemma~\ref{c1}. 
\end{proof}

Now we are ready to characterize unique key hypergraphs.

\begin{theorem}\label{t2}
For a Sperner hypergraph $\cB\subseteq 2^V$ the pure Horn function $h=\Phi_\cB$ is the only one with $\cK(h)=\cB$ if and only if for all 
$T\in\cB^d$ and $v\not\in T$ there exists $T'\in\cB^d$ such that $T'\neq T$ and $T'\subseteq T\cup\{v\}$.
\end{theorem}
\begin{proof}
For any pure Horn function $h$ with $\cK(h)=\cB$ we have $h\leq\Phi_\cB$.

For the only if direction, take an arbitrary $T\in\cB^d$ and $v\not\in
T$, and let $A=V\setminus (T\cup\{v\})$. By definition of \(A\), we have that
\(\cup\cB^{V\setminus A}\subseteq T\cup\{v\}\). Since \(T\) is a transversal,
we have that
\(T\subseteq\cup\cB^{V\setminus A}\) and that
\(A\to v\) is not an implicate of \(h\).
If \(v\not\in\cB^{V\setminus A}\), then by
Lemma~\ref{c1b} we have \(\cK(\Phi_\cB\land(A\to v))=\cB\) which is a
contradiction with the assumption that \(h\) is the only Horn function with this
property. It follows that \(v\in\cup\cB^{V\setminus A}\) and
altogether we get that $T\cup\{v\}=\cup\cB^{V\setminus A}$. 
In particular, this means that there
exists a $B\in\cB$ with $B\setminus A$ being minimal and $v\in B$.
Since $T$ is a transversal of $\cB$, we have $B\cap T\neq\emptyset$.
Consider an element $u\in B\cap T$. By the minimality of $B\setminus
A$, for every $B'\in\cB$ different from $B$ either
$B'\cap(T\setminus\{u\})\neq\emptyset$ or $v\in B'$. This means that $T'=(T\setminus\{u\})\cup\{v\}$ is
a transversal of $\cB$.

For the opposite direction, take an arbitrary $A\in\cB^*$ and $v\not\in\cup \cB^{V\setminus A}$. Then $A\cup\{v\}\in\cB^*$, hence there exists $T\in\cB^d$ disjoint from $A\cup\{v\}$. By the assumption, there exists $u\in T$ such that $T'=(T\setminus\{u\})\cup\{v\}$ is also a minimal transversal of $\cB$. Therefore there exists $B\in\cB$ for which $B\cap T'=\{v\}$. As $v\not\in\cup \cB^{V\setminus A}$, there exists $B'\in\cB$ such that $B'\setminus A\subsetneq B\setminus A$ and $v\not\in B'\setminus A$. This implies $B'\cap T'=\emptyset$, contradicting $T'$ being a transversal. This shows that the set $\Psi$ in Lemma~\ref{c1b} is empty, proving the uniqueness of $h$.
\end{proof}

We assume that the reader is familiar with the notion of a matroid \cite{nishimura2009lost,whitney1992abstract}.

\begin{corollary}
The cuts of a loopless matroid form a unique key hypergraph.
\end{corollary}
\begin{proof}
If $\cB$ is the set of cuts of a matroid, then $\cB^d$ is the set of bases. If the matroid is loopless, then $\cup\cB^d=V$. The basis exchange axiom implies the necessary and sufficient condition of Theorem~\ref{t2}.
\end{proof}

The following example shows that not all unique key hypergraphs are related to matroids. Let $\cB=\{12,13,14,234\}$, where $V=\{1,2,3,4\}$. Then $\cB^d=\cB$ and satisfies the conditions of Theorem~\ref{t2}, hence $\cB$ is unique key. Clearly, $\cB^d$ is not the set of bases of a matroid.

\begin{remark}
The conditions of Theorem~\ref{t2} can be checked in polynomial time if $\cB^d$ can be generated in (input) polynomial time from $\cB$. For example, if $\cB$ is $2$-monotone \cite{makino1997maximum} or forms the set of bases of a matroid. 
\end{remark}

\section{Unique Key Graphs}
\label{sec:ukg}

Let us now consider Sperner hypergraphs $\cB\subseteq 2^V$ such that $|B|=2$ for all $B\in \cB$ (i.e., graphs). For the sake of simplicity, we use $G=(V,E)$ to denote such a hypergraph $\cB=E$. We say that $G$ is a \emph{unique key graph} if $\cB=E$ is a unique key hypergraph. Following standard graph theory notation, we denote by $N(u)=\{v\in V\mid (u,v)\in E\}$ the set of neighbors of vertex $u\in V$. For a subset $S\subseteq V$ we denote by $N(S)=\left(\bigcup_{u\in S} N(u)\right)\setminus S$ the set of neighbors of $S$. 

\subsection{Complexity of Recognizing Unique Key Graphs}

Given a graph $G=(V,E)$ and a maximal independent set $I\subseteq V$ we say that $u\not\in I$ is an \emph{individual neighbor of $v\in I$} if $N(u)\cap I=\{v\}$. 

\begin{theorem}\label{c3}
A graph $G=(V,E)$ is unique key if and only if for every maximal independent set $I\subseteq V$ and vertex $v\in I$ there exists a vertex $u\not\in I$ that is an individual neighbor of $v$. 
\end{theorem}
\begin{proof}
The minimal transversals of $E$ are exactly the complements of the maximal independent sets of $G$, that is the minimal vertex covers of $G$. For a maximal independent set $I$ with $v\in I$ and $u\not\in I$, the set $(I\setminus\{v\})\cup\{u\}$ is an independent set if and only if $u$ is an individual neighbor of $v$. If this is the case, then $(I\setminus\{v\})\cup\{u\}$ can be extended to a maximal independent set $I'$ of $G$ not containing $v$. Thus the statement follows from Theorem~\ref{t2}.  
\end{proof}

Our next goal is to show that recognizing if $\cB$ is the set of minimal keys of a unique key function is difficult already for hypergraphs of dimension two. Let us consider a CNF $\Phi=C_1\wedge \dots \wedge C_m$ over Boolean variables $x_i$, $i=1,...,n$. Let us associate a graph $G_\Phi$ to $\Phi$ as follows. The set of vertices is $V(G_\Phi)=\{x_i,\bar{x}_i,y_i\mid i=1,...,n\}\cup\{C_j\mid j=1,...,m\}\cup\{z\}$. The edges are defined as follows: vertices $x_i$, $\bar{x}_i$ and $y_i$ form a triangle for all $i=1,...,n$. Vertices $C_j$, $j=1,...,m$ and $z$ form a clique. Finally, all vertices $C_j$ are connected to the literals they include (see Figure~\ref{fig:np}).

\begin{theorem}\label{t3}
A CNF $\Phi$ is not satisfiable if and only if the graph $G_\Phi$ is unique key. 
\end{theorem}
\begin{proof}
We derive this claim using Theorem \ref{c3}.

Let us note first that every maximal independent set $I\subseteq V(G_\Phi)$ has exactly $n+1$ points, one from each of the following cliques: $T_i=\{x_i,\bar{x}_i,y_i\}$, $i=1,...,n$, and $K=\{C_j\mid j=1,...,m\}\cup\{z\}$. This is because an independent set $I$ can contain at most one vertex from each of these cliques, and if it is disjoint from $T_i$, then $I\cup\{y_i\}$ is also independent. Similarly, if $I\cap K=\emptyset$, then $I\cup\{z\}$ is also independent. We now verify the conditions of Theorem \ref{c3}.

Let $I$ be a maximal independent set, and assume that $v=x_i\in I$ or $v=\bar{x}_i\in I$. In both cases $u=y_i$ is an individual neighbor of $v$. Note next that the sets $N(x_i)\cap K$ and $N(\bar{x}_i)\cap K$ are disjoint, and therefore any independent set is disjoint from at least one of these sets. 
Thus, if $v=y_i\in I$, then either $u=x_i$ or $u=\bar{x}_i$ (or both) is an individual neighbor of $v$. 
If $v=C_j\in I$, then $u=z$ is an individual neighbor of $v$. 

Thus, the only claim left to show is that $\Phi$ is satisfiable if and only if there exists a maximal independent set $I$ of $G$ containing vertex $z$  such that $z$ does not have an individual neighbor. To see this let us first assume that $\Phi$ is satisfiable. Consider the set $I$ that contains the literals that are true in a satisfying assignment and vertex $z$. Since every clause $C_j$ is satisfied, it has a neighbor in $I$ other than $z$, and thus $z$ does not have an individual neighbor. For the other direction let us assume that $I$ is a maximal independent set, containing $z$ such that $z$ does not have an individual neighbor. Therefore, every clause $C_j$ must have a neighbor in $I$, which must be a literal. Since $(x_i,\bar{x}_i)$ is an edge of $G$ for all $i=1,...,n$,  $I$ cannot contain a complementary pair of literals, and thus the literals in $I$ can be set to true simultaneously, satisfying $\Phi$. 
\end{proof}

\begin{figure}[t!]
  \centering
  \includegraphics[width=0.6\linewidth]{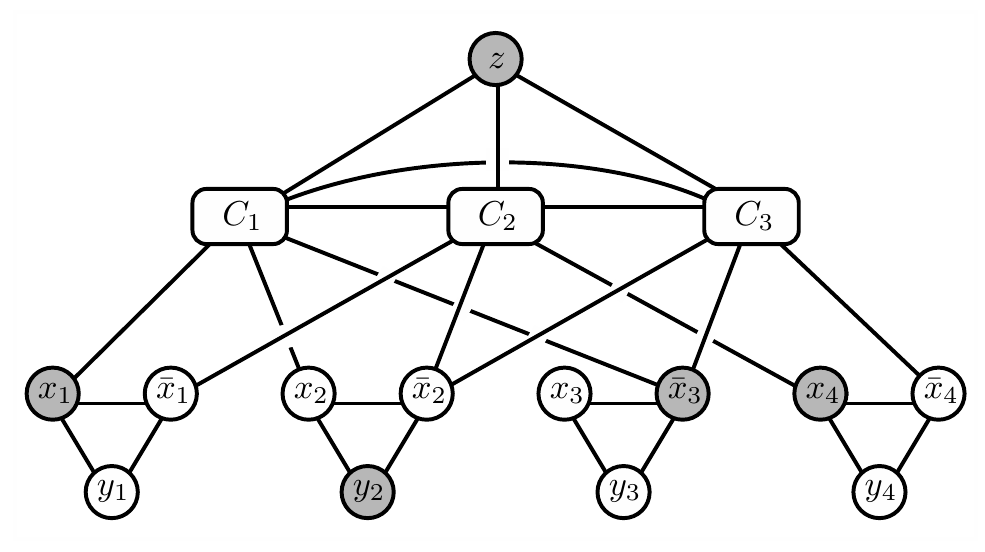}
  \caption{The graph $G_\Phi$ corresponding to CNF formula $\Phi=(x_1\vee x_2\vee\bar{x}_3)\wedge(\bar{x}_1\vee\bar{x}_2\vee x_4)\wedge(\bar{x_2}\vee\bar{x}_3\vee\bar{x}_4)$. Grey vertices form a maximal independent set corresponding to a satisfying truth assignment. Note that $z$ has no individual neighbor.}
\label{fig:np}
\end{figure}

\begin{corollary}\label{c4}
Deciding if a hypergraph is unique key is co-$\mathsf{NP}$-complete already for hypergraphs of dimension 2.
\end{corollary}
\begin{proof}
It is easy to see that the problem belongs to co-$\mathsf{NP}$, and thus the statement follows by Theorem~\ref{t3}.
\end{proof}

\subsection{Bipartite Graphs}

\begin{theorem}
A bipartite graph $G=(V,E)$ without isolated vertices is unique key if and only if $E$ is a perfect matching.  
\end{theorem}
\begin{proof}
If $E$ forms a perfect matching on $V$, then every maximal independent set $I$ contains exactly one end vertex of every edge in $E$. For any vertex $v\in I$, the other end vertex $u$ of the matching edge incident to $v$ is an individual neighbor of $v$, thus $G$ is unique key by Theorem~\ref{c3}.

For the other direction, let $A$ and $B$ denote the color classes of $G$, that is, $V=A\cup B$. By the assumption that there are no isolated vertices in $G$, both $A$ and $B$ are maximal independent sets. By Theorem~\ref{c3}, every vertex $v\in V$ has an individual neighbor in the opposite color class, that is, a neighbor of degree exactly one. This implies that $E$ is a matching as stated.  
\end{proof}

\subsection{Bounded Treewidth Graphs}

\begin{theorem}
   \label{thm:treewidth}
   For graphs of bounded treewidth, it is possible to decide in linear
   time if a graph is a unique key graph.
\end{theorem}
\begin{proof}
   We will formulate the problem in monadic second order logic (MSO),
   the result then follows by Courcelle's
   theorem~\cite{courcelle1990}. Assume that a graph \(G=(V, E)\) is
   described with a set of vertices \(V\) and an adjacency relation
   \(\operatorname{adj}(u, v)\) which represents the set of edges. The
   unique key property can then be described as the predicate
   \begin{equation*}
      \operatorname{UniqKey}(G)=(\forall I\subseteq V)(\forall v\in I)(\exists u\in V)[\operatorname{IndSet}(I)
      \to \operatorname{IndNeigh}(I, v, u)]
   \end{equation*}
   where \(\operatorname{IndSet(I)}\) is a predicate satisfied if
   \(I\) is an independent set of \(G\) and
   \(\operatorname{IndNeigh}(I, v, u)\) is satisfied if \(v\in I\) and \(u\) is its
   individual neighbour. These predicates can be defined in the
   following way.
   \begin{align*}
      \operatorname{IndSet}(I)&=(\forall u\in I)(\forall v\in I)[\neg \operatorname{adj}(u, v)]\\
      \operatorname{IndNeigh}(I, v, u)&=(\forall w\in I)[\operatorname{adj}(w, u)\to w=v]\qedhere
   \end{align*}
\end{proof}

Since the formulation of \(\operatorname{UniqKey(G)}\) uses only
quantification over a set of vertices \(I\) and not over any set of
edges, we can use it to show the following corollary.

\begin{corollary}
   For graphs of bounded clique-width, it is possible to decide in
   linear time if a graph is a unique key graph.
\end{corollary}
\begin{proof}
   Follows by using a version of Courcelle's theorem for
   clique-width~\cite{courcelle2000} on the formulation of
   predicate~\(\operatorname{UniqKey}(G)\) in the proof of
   Theorem~\ref{thm:treewidth}.
\end{proof}

\subsection{Graphs With Small Induced Matchings}

\begin{theorem}
Let $G=(V,E)$ be a graph, and assume that the size of the largest induced matching of $G$ is bounded by a constant. Then there is an efficient algorithm to decide if $G$ is a unique key graph.
\end{theorem}
\begin{proof}
If $\cB=E$ then $\cB^d$ is the family of minimal vertex covers that are exactly the complements of maximal independent sets. It is known that if the largest induced matching in $G$ has size at most $p$, then it has at most $n^{2p}$ maximal independent sets \cite{balas1989graphs}. Thus if $p$ is a constant, then all of them can be generated in polynomial time \cite{tsukiyama1977new}. This in turn implies that the conditions of Theorem~\ref{c3} can be checked in polynomial time.
\end{proof}

\section{Generating Minimal Keys}
\label{sec:lex}

We shift the focus from unique key hypergraphs to the problem of generating all possible minimal keys of a given pure Horn function.  The proposed approach can be applied for various problems, for example for generating all \emph{minimal target sets} of a graph. Note that the number of minimal keys can be exponential in the size of the input CNF, hence the efficiency of generating them is measured by the time spent between outputting two of them. A \emph{generation algorithm} outputs the objects in question one by one without repetition. Such a procedure is called \emph{polynomial delay} if the computing time between any two consecutive outputs is bounded by a polynomial of the input size.

Given a pure Horn CNF $\Phi$, we associate to it a directed graph $D_\Phi=(\cK(\Phi),E)$ as follows. For a minimal key $K\in\cK(\Phi)$, an arbitrary variable $v\in K$, and a clause $A\to v\in\Phi$, we define the set $S=(K-v)\cup A$. Note that $S$ is a key of $\Phi$, hence there exists $K'\in\cK(\Phi)$ with $K'\subseteq S$. We find such a $K'$ using a greedy procedure by dropping variables from $S$ one-by-one, and checking at each step  if the remaining set is a key by using forward chaining with respect to $\Phi$. We include the directed edge $KK'$ into $E$ for all possible choices $v\in K$ and $A\to v\in\Phi$. For some $v\in K$ we might not have a clause $A\to v$ in which case we do not generate the corresponding $K'$. Note that every vertex in $D_\Phi$ has at most $m$ outgoing edges. Let us remark that the final graph $D_\Phi$ is not uniquely defined as its edge set depends on the choices of the $K'$ sets in the above procedure.

\begin{lemma}\label{lem:sc}
$D_\Phi$ is strongly connected.
\end{lemma}
\begin{proof}
First we introduce a measure between minimal keys. Let $K_1,K_2\in\cK(\Phi)$ be two minimal keys. We know that the forward chaining closure of $K_2$ is equal to $V$. Let us partition $V$ into layers $L_0,L_1,\dots,L_t$ where $L_0:=K_2$, define $L_{i+1}:=\{v\in V\setminus L_i\mid \text{there exists}\ A\to v\in\Phi\ \text{s.t.}\ A\subseteq \bigcup_{j=0}^i L_j\}$, and $t$ is the largest index such that $L_t\neq\emptyset$. Let $\varrho(K_1,K_2):=(\varrho_0,\varrho_1,\dots,\varrho_t)$ where $\varrho_i=|L_i\cap K_1|$ for $i=0,\dots,t$.

We claim that there exists an out-neighbor $K_3$ of $K_1$ in $D_\Phi$ such that $\varrho(K_3,K_2)$ is strictly smaller in the reverse lexicographic order than $\varrho(K_1,K_2)$. To see this, let $i$ be the largest index such that $\varrho_i\neq 0$, and let $v$ be in $K_1\cap L_i$. Since $v\in L_i$, there exists an $A\to v\in\Phi$ such that $A\subseteq \bigcup_{j=0}^{i-1} L_j$. For the set $S=(K_1-v)\cup A$ we have that $|L_i\cap S|<|L_i\cap K_1|$ and $|L_j\cap S|=0$ for $j>i$. Thus the out-neighbor $K_3\subseteq S$ satisfies the claim. By induction in the reverse lexicographic order of the possible $\varrho$ vectors, there exists a directed path in $D_\Phi$ from $K_3$ to $K_2$. As $K_1K_3\in E$, the same holds for $K_1$, thus finishing the proof of the lemma.
\end{proof}

Next we propose an algorithm similar to the approach used in \cite{boros1990polynomial,eiter2007computing} for generating all prime implicates and all abductive explanations of a Horn CNF.

\begin{theorem} \label{thm:genkeys}
Given a pure Horn CNF $\Phi$, we can generate all minimal keys of $\Phi$ with polynomial delay. 
\end{theorem}
\begin{proof}
Consider the directed graph $D_\Phi$. Our algorithm will generate all out-neighbors of the minimal keys that are already generated, starting from a minimal key which we generate by greedily leaving out elements from $V$. As $D_\Phi$ is strongly connected according to Lemma~\ref{lem:sc}, all minimal keys are obtained this way.

The set of minimal keys that are already generated is kept in a last-in-first-out queue. Before outputting the top element of the queue, we generate all its out-neighbors and add the new ones to the queue. Since the generation of the out-neighbors can be done in polynomial time the numbers of variables and clauses, this procedure has a polynomial delay.
\end{proof}

\subsection{Minimum Target Set Selection}

In the {\sc Minimum Target Set Selection} problem, we are given an undirected graph $G=(V,E)$ and a threshold function $t:V\to\mathbb{Z}_+$. As a starting step, we can activate a subset $S\subseteq V$ of vertices. In every subsequent round, a vertex $v$ becomes activated if at least $t(v)$ of its neighbors are already active. The goal is to find a minimum sized initial set $S$ of active nodes (called a \emph{target set}) so that the activation spreads to the entire graph.

Finding a minimum sized target set is rather difficult. Chen \cite{chen2009approximability} showed that the problem is difficult to approximate within a $O(\poly\log(n))$ factor already when all thresholds are 2 and the graph has a constant degree. Charikar et al. \cite{charikar2016approximating} proved that, assuming the Planted Dense Subgraph conjecture, {\sc Minimum Target Set Selection} is in fact difficult to approximate within a factor of $O(n^{1/2-\varepsilon})$ for every $\varepsilon>0$ even for constant thresholds.

The aim of this section is to show that the problems of finding a minimum target set in a graph ({\sc Min-TSS}) and of finding a minimum key of a pure Horn function ({\sc Min-Key}) are closely related. 

\begin{figure}
\centering
\begin{subfigure}[t]{.48\textwidth}
  \centering
  \includegraphics[width=.6\linewidth]{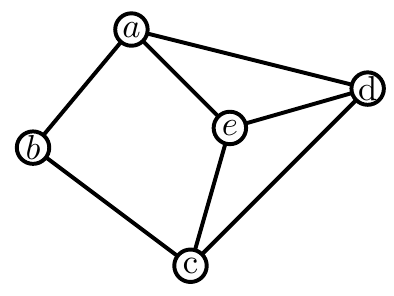}
  \caption{Instance of \textsc{Min-TSS} problem. The thresholds are $t(a)=t(b)=t(c)=t(d)=1$ and $t(e)=2$.}
  \label{fig:red1}
\end{subfigure}%
\hfill
\begin{subfigure}[t]{.48\textwidth}
  \centering
  \includegraphics[width=.6\linewidth]{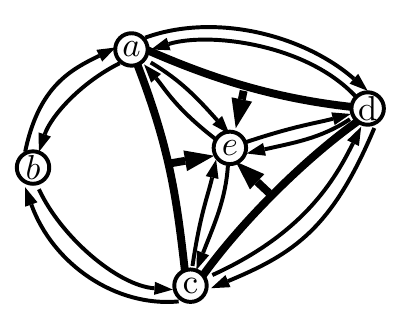}
  \caption{Construction of $\Psi_G$. Thick hyperedges represent clauses containing three variables.}
  \label{fig:red2}
\end{subfigure}
\caption{Illustration of Theorem~\ref{thm:red1}. The CNF associated to $G$ is $\Psi_G=(b\to a)\wedge(e\to a)\wedge(d\to a)\wedge(a\to b)\wedge(c\to b)\wedge(b\to c)\wedge(d\to c)\wedge(e\to c)\wedge (a\to d)\wedge(c\to d)\wedge(e\to d)\wedge(\{a,c\}\to e)\wedge(\{a,d\}\to e)\wedge(\{c,d\}\to e)$.}
\label{fig:reda}
\end{figure}

\begin{theorem}\label{thm:red1}
The \textsc{Min-TSS} problem with constant thresholds is polynomial-time reducible to the \textsc{Min-Key} problem.
\end{theorem}
\begin{proof}
Let $G=(V,E)$, $t:V\to\mathbb{Z}_+$ be an instance of the \textsc{Min-TSS} problem. For a vertex $v\in V$, we denote the set of its neighbors by $N(v)\subseteq V$. We construct a Horn CNF as follows (see Figure~\ref{fig:reda}):
\begin{equation*}
\Psi_G:=\bigwedge_{v\in V}\bigwedge_{\substack{A\subseteq N(v)\\ |A|=t(v)}} A\to v.
\end{equation*}
Note that $\Psi_G$ can be determined in polynomial time as the thresholds are assumed to be constants. By the definition of $\Phi_G$, the activation process in $G$ is equivalent to the forward chaining process in $\Psi_G$. This means that $K\subseteq V$ is a target set of $G$ if and only if it is a key of $\Psi_G$, concluding the proof of the theorem.
\end{proof}

Theorem~\ref{thm:red1} together with the hardness result of \cite{charikar2016approximating} implies that \textsc{Min-Key} is difficult to approximate within a factor of $O(n^{1/2-\varepsilon})$ for every $\varepsilon>0$, assuming that the Planted Dense Subgraph conjecture holds.

Based on a construction previously used in \cite{charikar2016approximating} for establishing a connection between the directed and undirected variants of the target set selection problem, we show that \textsc{Min-TSS} includes \textsc{Min-Key} as a special case.

\begin{theorem}\label{thm:red2}
The \textsc{Min-Key} problem is polynomial-time reducible to the \textsc{Min-TSS} problem.
\end{theorem}
\begin{proof}
Let $\Phi$ be a pure Horn CNF on variables $V$. We construct a graph $G=(V',E)$ together with a threshold function $t:V\to\mathbb{Z}_+$ such that every key of $\Phi$ is a target set of $G$, while every target set of $G$ can be transformed to a key of $\Phi$ without increasing the size of the set. 

We add the set of variables $V$ to the vertices of $G$, and define $t(v)=1$ for $v\in V$. For every clause $C=A\to v$ of $\Phi$, we construct a gadget as follows. We add a vertex $p^C$ that corresponds to the clause and set $t(p^C)=|A|$. For every variable $a\in A$, we add four new vertices $x^C_a,y^C_a,z^C_a$ and $w^C_a$ with thresholds $t(x^C_a)=t(y^C_a)=t(z^C_a)=1$ and $t(w^C_a)=2$, together with the edges $ax^C_a,x^C_ay^C_a,x^C_az^C_a,y^C_aw^C_a,z^C_aw^C_a$ and $w^C_ap^C$. Finally, we add four new vertices $x^C_v,y^C_v,z^C_v$ and $w^C_v$ with thresholds $t(x^C_v)=t(y^C_v)=t(z^C_v)=1$ and $t(w^C_v)=2$, together with the edges $p^Cx^C_v,x^C_vy^C_v,x^C_vz^C_v,y^C_vw^C_v,z^C_vw^C_v$ and $w^C_vv$ (see Figure~\ref{fig:redb}). 

If $K\subseteq V$ is a key of $\Phi$, then the same set of vertices in $G$ form a target set. Indeed, when the forward chaining procedure uses a clause $C=A\rightarrow v$ to reach a variable $v$, then this step corresponds to the activation of $v$ through the gadget associated to $C$ in $G$. 

Now assume that $S$ is a target set of $G$. We cannot directly say that $S$ is a key of $\Phi$ as $S$ might contain vertices from $V'\setminus V$. However, it is not difficult to see that 
\begin{align*}
K
{}&{}:=
(V\cap S)\\
{}&{}~~~~\cup
\{v\in V\mid \text{there exists}\ C\in\Phi\ \text{with}\ v\in C,S\cap\{x^C_v,y^C_v,z^C_v,w^C_v\}\neq\emptyset\}\\
{}&{}~~~~\cup
\{v\in V\mid\text{there exists}\ C=A\to v\in\Phi\ \text{with}\ p_C\in S\}
\end{align*}
is a key of $\Phi$ with $|K|\leq |S|$, concluding the proof of the theorem. 
\end{proof}

\begin{figure}
\centering
\begin{subfigure}[t]{.48\textwidth}
  \centering
  \includegraphics[width=.6\linewidth]{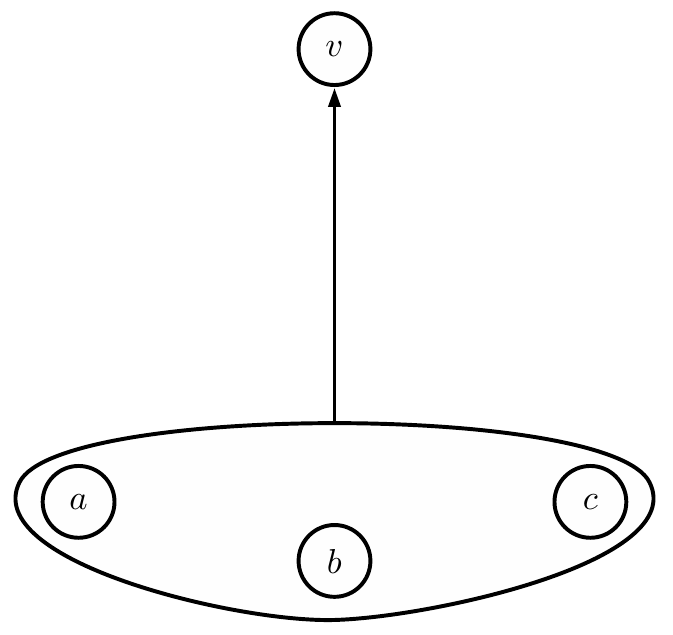}
  \caption{A pure Horn clause $C=A\to v$, where $A=\{a,b,c\}$.}
  \label{fig:red3}
\end{subfigure}%
\hfill
\begin{subfigure}[t]{.48\textwidth}
  \centering
  \includegraphics[width=.6\linewidth]{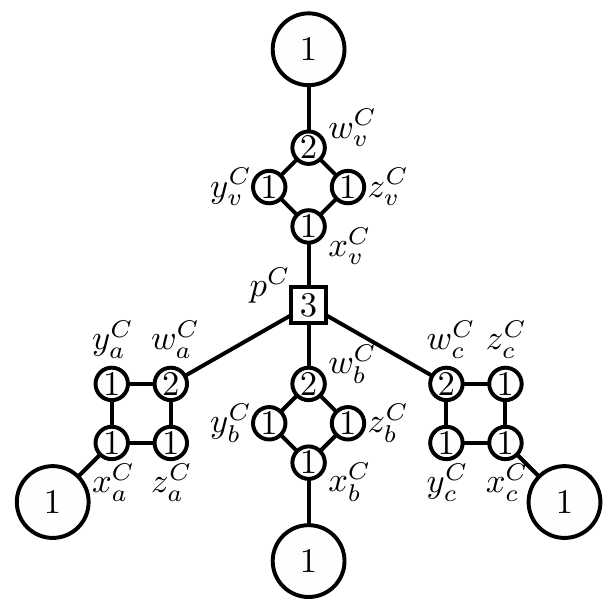}
  \caption{The gadget and threshold values corresponding to $C=A\to v$.}
  \label{fig:red4}
\end{subfigure}
\caption{Illustration of Theorem~\ref{thm:red2}. Note that the size of the graph $G$ is polynomial in the length of the input.}
\label{fig:redb}
\end{figure}

We have seen that finding a minimum sized target set is difficult already for constant thresholds. However, by combining Theorems~\ref{thm:genkeys} and \ref{thm:red1}, we get the following result.

\begin{corollary}
Given a graph $G=(V,E)$ and constant thresholds $t:V\rightarrow\mathbb{Z}_+$, we can generate all minimal target sets of $G$ with polynomial delay.
\end{corollary}

\section{Conclusions}

In this paper we defined unique key hypergraphs as Sperner hypergraphs that form the set of minimal keys of a unique pure Horn function. We gave a characterization of such hypergraphs, and showed that cuts of a matroid form a natural example. We proved that the recognition of unique key hypergraphs is co-$\mathsf{NP}$-complete already when every hyperedge has size two. We identified several classes of graphs for which the recognition problem can be decided in polynomial time. We gave an algorithm for generating all minimal keys of a pure Horn function with polynomial delay. By showing that the problems of finding a minimum key of a pure Horn function and of finding a minimum target set of a graph are closely related, we extended our algorithm to generate all minimal target sets of a graph with polynomial delay when the thresholds are bounded by a constant. It remains an open question whether all minimal target sets can be generated with polynomial delay when the thresholds are unbounded.

\medskip
\paragraph{Acknowledgements} Krist\'of B\'erczi was supported by the J\'anos Bolyai Research Fellowship of the Hungarian Academy of Sciences and by the ÚNKP-19-4 New National Excellence Program of the Ministry for Innovation and Technology. Ond\v{r}ej \v{C}epek and Petr Ku\v{c}era gratefully acknowledge a support by the Czech Science Foundation (Grant 19-19463S). Projects no. NKFI-128673 and no. ED\_18-1-2019-0030 (Application-specific highly reliable IT solutions) have been implemented with the support provided from the National Research, Development and Innovation Fund of Hungary, financed under the FK\_18 and the Thematic Excellence Programme funding schemes, respectively. This  work  was  supported  by  the  Research  Institute  for  Mathematical  Sciences,  an  International Joint Usage/Research Center located in Kyoto University. 

\bibliographystyle{abbrv}
\bibliography{min_keys}

\begin{thebibliography}{10}

\bibitem{AB09}
M.~Arias and J.~L. Balc{\'a}zar.
\newblock Canonical {H}orn representations and query learning.
\newblock In {\em International Conference on Algorithmic Learning Theory},
  pages 156--170. Springer, 2009.

\bibitem{AB11}
M.~Arias and J.~L. Balc{\'a}zar.
\newblock Construction and learnability of canonical {H}orn formulas.
\newblock {\em Machine Learning}, 85(3):273--297, 2011.

\bibitem{A74}
W.~W. Armstrong.
\newblock {\em Dependency structures of database relationships}.
\newblock Proc. IFIP 74. North Holland, Amsterdam, pp. 580-583, 1974.

\bibitem{ADS86}
G.~Ausiello, A.~D'Atri, and D.~Sacca.
\newblock Minimal representation of directed hypergraphs.
\newblock {\em SIAM Journal on Computing}, 15(2):418--431, 1986.

\bibitem{balas1989graphs}
E.~Balas and C.~S. Yu.
\newblock On graphs with polynomially solvable maximum-weight clique problem.
\newblock {\em Networks}, 19(2):247--253, 1989.

\bibitem{arxiv}
K.~{B{\'e}rczi}, E.~{Boros}, O.~{{\v C}epek}, P.~{Ku{\v c}era}, and
  K.~{Makino}.
\newblock {Approximating minimum representations of key {H}orn functions}.
\newblock {\em ArXiv e-prints}, Nov. 2018.

\bibitem{boros1990polynomial}
E.~Boros, Y.~Crama, and P.~L. Hammer.
\newblock Polynomial-time inference of all valid implications for horn and
  related formulae.
\newblock {\em Annals of Mathematics and Artificial Intelligence},
  1(1-4):21--32, 1990.

\bibitem{CASPARD2003241}
N.~Caspard and B.~Monjardet.
\newblock The lattices of closure systems, closure operators, and implicational
  systems on a finite set: a survey.
\newblock {\em Discrete Applied Mathematics}, 127(2):241 -- 269, 2003.
\newblock Ordinal and Symbolic Data Analysis (OSDA '98), Univ. of
  Massachusetts, Amherst, Sept. 28-30, 1998.

\bibitem{charikar2016approximating}
M.~Charikar, Y.~Naamad, and A.~Wirth.
\newblock On approximating target set selection.
\newblock In {\em Approximation, Randomization, and Combinatorial Optimization.
  Algorithms and Techniques (APPROX/RANDOM 2016)}. Schloss
  Dagstuhl-Leibniz-Zentrum fuer Informatik, 2016.

\bibitem{chen2009approximability}
N.~Chen.
\newblock On the approximability of influence in social networks.
\newblock {\em SIAM Journal on Discrete Mathematics}, 23(3):1400--1415, 2009.

\bibitem{courcelle1990}
B.~Courcelle.
\newblock The monadic second-order logic of graphs. i. recognizable sets of
  finite graphs.
\newblock {\em Information and Computation}, 85(1):12 -- 75, 1990.

\bibitem{courcelle2000}
B.~Courcelle, J.~A. Makowsky, and U.~Rotics.
\newblock Linear time solvable optimization problems on graphs of bounded
  clique-width.
\newblock {\em Theory of Computing Systems}, 33(2):125--150, Apr 2000.

\bibitem{DOWLING1984267}
W.~F. Dowling and J.~H. Gallier.
\newblock Linear-time algorithms for testing the satisfiability of
  propositional {H}orn formulae.
\newblock {\em The Journal of Logic Programming}, 1(3):267 -- 284, 1984.

\bibitem{eiter2007computing}
T.~Eiter and K.~Makino.
\newblock On computing all abductive explanations from a propositional horn
  theory.
\newblock {\em Journal of the ACM (JACM)}, 54(5):24--es, 2007.

\bibitem{GD86}
J.-L. Guigues and V.~Duquenne.
\newblock Familles minimales d'implications informatives r{\'e}sultant d'un
  tableau de donn{\'e}es binaires.
\newblock {\em Math{\'e}matiques et Sciences humaines}, 95:5--18, 1986.

\bibitem{HK93}
P.~L. Hammer and A.~Kogan.
\newblock Optimal compression of propositional {H}orn knowledge bases:
  complexity and approximation.
\newblock {\em Artificial Intelligence}, 64(1):131--145, 1993.

\bibitem{Kucera2014HydrasCO}
P.~Ku\v{c}era.
\newblock Hydras: Complexity on general graphs and a subclass of trees.
\newblock {\em Theor. Comput. Sci.}, 658:399--416, 2014.

\bibitem{maier}
D.~Maier.
\newblock Minimum covers in the relational database model.
\newblock In {\em Proceedings of the eleventh annual ACM symposium on Theory of
  computing}, pages 330--337. ACM, 1979.

\bibitem{makino1997maximum}
K.~Makino and T.~Ibaraki.
\newblock The maximum latency and identification of positive boolean functions.
\newblock {\em SIAM Journal on Computing}, 26(5):1363--1383, 1997.

\bibitem{nishimura2009lost}
H.~Nishimura and S.~Kuroda.
\newblock {\em A Lost Mathematician, {T}akeo {N}akasawa: The Forgotten Father
  of Matroid Theory}.
\newblock Springer Science \& Business Media, 2009.

\bibitem{seymour1973incomparable}
P.~Seymour.
\newblock On incomparable families of sets.
\newblock {\em Mathematica}, 20:208--209, 1973.

\bibitem{Sloan2017HydrasDH}
R.~H. Sloan, D.~Stasi, and G.~Tur{\'a}n.
\newblock Hydras: Directed hypergraphs and {H}orn formulas.
\newblock {\em Theor. Comput. Sci.}, 658:417--428, 2017.

\bibitem{tsukiyama1977new}
S.~Tsukiyama, M.~Ide, H.~Ariyoshi, and I.~Shirakawa.
\newblock A new algorithm for generating all the maximal independent sets.
\newblock {\em SIAM Journal on Computing}, 6(3):505--517, 1977.

\bibitem{whitney1992abstract}
H.~Whitney.
\newblock On the abstract properties of linear dependence.
\newblock In {\em Hassler Whitney Collected Papers}, pages 147--171. Springer,
  1992.

\end{thebibliography}

\end{document}